\documentclass[12pt]{amsart}
\usepackage{txfonts}
\usepackage{amssymb}
\usepackage{eucal}
\usepackage{graphicx}
\usepackage{amsmath}
\usepackage{amscd}
\usepackage[all]{xy}           

\usepackage{amsfonts,latexsym}
\usepackage{xspace}
\usepackage{epsfig}
\usepackage{float}
\usepackage{color}
\usepackage{fancybox}
\usepackage{colordvi}
\usepackage{multicol}
\usepackage{colordvi}
\usepackage{stmaryrd}
\usepackage[colorlinks,final,backref=page,hyperindex,hypertex]{hyperref}

\newtheorem{theorem}{Theorem}[section]
\newtheorem{prop}[theorem]{Proposition}
\newtheorem{defn}[theorem]{Definition}
\newtheorem{lemma}[theorem]{Lemma}
\newtheorem{coro}[theorem]{Corollary}
\newtheorem{prop-def}{Proposition-Definition}[section]

\newtheorem{remark}[theorem]{Remark}

\newcommand{\nc}{\newcommand}
\newcommand{\delete}[1]{}

\nc{\mlabel}[1]{\label{#1}}  
\nc{\mcite}[1]{\cite{#1}}  
\nc{\mref}[1]{\ref{#1}}  
\nc{\mbibitem}[1]{\bibitem{#1}} 

\nc{\bfk}{\mathbf{k}}
\nc{\Der}{\mathrm{Der}}
\nc{\Ker}{\mathrm{Ker}}

\begin{document}

\title{Rota-Baxter $3$-Lie algebras}

\author{RuiPu  Bai}
\address{College of Mathematics and Computer Science,
Hebei University, Baoding 071002, China}
\email{bairp1@yahoo.com.cn}

\author{Li Guo}
\address{
    Department of Mathematics and Computer Science,
         Rutgers University,
         Newark, NJ 07102, USA}
\email{liguo@rutgers.edu}

\author{Jiaqian  Li}
\address{College of Mathematics and Computer Science,
Hebei University, Baoding 071002, China}
\email{lijiaqianjiayou@163.com}

\author{Yong Wu}
\address{College of Mathematics and Computer Science,
Hebei University, Baoding 071002, China}
\email{wuyg1022@sina.com}

\date{\today}

\begin{abstract} In this paper we introduce the concepts of a Rota-Baxter operator and a differential operator
with weights on an $n$-algebra. We then focus on Rota-Baxter $3$-Lie algebras and show that they can be derived from
Rota-Baxter Lie algebras and pre-Lie algebras and from Rota-Baxter commutative associative algebras with derivations.
We also establish the inheritance property of Rota-Baxter $3$-Lie algebras.
\end{abstract}

\subjclass[2010]{17B05, 17D99.}

\keywords{ Rota-Baxter operator, Rota-Baxter $n$-algebra,
Rota-Baxter  $n$-Lie algebra, differential $n$-algebra, pre-Lie algebra, Lie triple system.}

\maketitle



\allowdisplaybreaks

\section{Introduction}

 $n$-Lie algebras~\cite{F} are a type of multiple  algebraic
systems appearing in many fields of mathematics and mathematical
physics~\cite{N, T, BL,HHM,HCK,G,P}. For example, the structure
of $3$-Lie algebras is applied to the study of the supersymmetry and
gauge symmetry transformations of the world-volume theory of
multiple coincident M2-branes; the Bagger-Lambert theory has a novel
local gauge symmetry which is based on a metric $3$-Lie algebra; the
identity in Eq.~(\mref{eq:nlie}) for a $3$-Lie algebra is essential to define the action
with $N=8$ supersymmetry; the $n$-Jacobi identity can be regarded as
a generalized Plucker relation in the physics literature, and so on.
The theory of $n$-Lie algebras has been widely studied~\cite{K,L,BSZ,BSZ1, BBW, BHB, AI}.

P. Ho, Y. Imamura and Y. Matsuo in paper \cite{HIM} studied two
derivations of the multiple $D2$ action from the multiple $M2$-brane
model proposed by Bagger-Lambert and Gustavsson. The first one is to
start from  $3$-Lie algebras given by arbitrary Lie algebras through
$2$-dimensional extensions.
 The first author and collaborators~\cite{BBW} realized $3$-Lie algebras by Lie
algebras and several linear functions. Filippov~\cite{F} constructed $(n-1)$-Lie algebras from $n$-Lie algebras by
fixing some non-zero vector in the multiplication. We are motivated
by this to construct $n$-Lie algebras.
However, it is not easy here due to the $n$-ary operation.
In order to avoid the complicated equations involving the structural
constants like in the classification~\cite{BSZ1}, it is
natural to consider to use Lie algebras and other better studied algebras to obtain $3$-Lie and $n$-Lie algebras.

In recent years, Rota-Baxter (associative) algebras, originated from the work of G. Baxter~\mcite{Ba} in probability and populated by
the work of Cartier and Rota~\mcite{Ca,Ro1,Ro2}, have also been studied in connection with many areas of mathematics
and physics, including combinatorics, number theory, operads and quantum field theory~\cite{Ag2,BBGN,EGK,Guw,Gub,GK1,GSZ,GZ,Ro1,Ro2}.
In particular Rota-Baxter algebras have played an important role in the Hopf algebra approach of renormalization of perturbative quantum field theory of Connes and Kreimer~\mcite{CK,EGK,EGM}, as well as in the application of the renormalization method in solving divergent problems in number theory~\mcite{GZ,MP}. Furthermore, Rota-Baxter operators on a Lie algebra are an operator form of the classical Yang-Baxter equations and contribute to the study of integrable systems~\cite{Bai,BGN,BGN2,STS}. Further Rota-Baxter Lie algebras are closely related to pre-Lie algebras and PostLie algebras.

Thus it is time to study $n$-Lie algebras and Rota-Baxter algebras together to get a suitable definition of Rota-Baxter $n$-Lie algebras. In this paper we investigate Rota-Baxter $n$-algebras in the context of associative and Lie algebras with focus on Rota-Baxter $3$-Lie algebras.
We establish a close relationship of our definition of Rota-Baxter $3$-Lie algebras with well-known concepts of Rota-Baxter associative, commutative or Lie algebras. This on one hand justifies the definition of Rota-Baxter $3$-Lie algebras and
and on the other hand provides a rich source of examples for Rota-Baxter $3$-Lie algebras.
The concepts of differential operators and Rota-Baxter operators with weights for general
(non-associative) algebras are introduced in Section~\mref{sec:rbn}. The duality of the two concepts are established.
In Section~\ref{sec:rbl3rbl} we extend the connections~\mcite{HuiB1,BBW,HuiB}
from Lie algebras and pre-Lie algebras to $3$-Lie algebras to the context of Rota-Baxter $3$-Lie algebras. In Section~\mref{sec:com3lie},
we construct Rota-Baxter $3$-Lie algebras from commutative (associative) Rota-Baxter algebras together with commuting derivations and suitable
linear forms. In Section~\mref{sec:rb3inh} we study the inheritance property of Rota-Baxter $3$-Lie algebras as in the case of Rota-Baxter Lie algebras. We also consider the refined case of Lie triple systems.

\section{Differential $n$-algebras and Rota-Baxter $n$-Lie algebras}
\mlabel{sec:rbn}

An {\bf $n$-Lie algebra} is a vector space $\frak g$ over a field $\bfk$ endowed with an $n$-ary multi-linear skew-symmetric operation
$[x_1, \cdots, x_n]$ satisfying the $n$-Jacobi identity
\begin{equation}
  [[x_1, \cdots, x_n], y_2, \cdots, y_n]=\sum_{i=1}^n[x_1, \cdots, [ x_i, y_2, \cdots, y_n], \cdots,
  x_n].
\mlabel{eq:nlie}
\end{equation}
In particular, a {\bf 3-Lie algebra} is a vector space $\frak g$ endowed with a ternary multi-linear skew-symmetric operation
\begin{equation}
[[x_1,x_2,x_3],y_2,y_3]=[[x_1,y_2,y_3],x_2,x_3] +[x_1,[x_2,y_2,y_3],x_3]+[x_1,x_2,[x_3,y_2,y_3]],
\end{equation}
for all $x_1,x_2,x_3, y_2, y_3\in \frak g$.
Under the skew-symmetric condition, the equation is equivalent to
\begin{equation}
[[x_1,x_2,x_3],y_2,y_3]=[[x_1,y_2,y_3],x_2,x_3] +[[x_2,y_2,y_3],x_3,x_1]+[[x_3,y_2,y_3],x_1,x_2].
\end{equation}

Let $(A, \cdot)$ be a $\bfk$-vector space with a binary operation $\cdot$ and let $\lambda\in \bfk$. If a linear map $P: A\rightarrow A$ satisfies
\begin{equation}
P(x)\cdot P(y)=P(P(x)\cdot y+x\cdot P(y)+\lambda x\cdot y) \ \text{ for all } x, y \in A,
\mlabel{eq:rb}
\end{equation}
then $P$ is called a {\bf Rota-Baxter operator of weight $\lambda$} and $(A, \cdot, P)$ is called a {Rota-Baxter
algebra of weight $\lambda$}.

If a linear map $D:A\to A$ satisfies
\begin{equation}
D(x\cdot y)=D(x)\cdot y+x\cdot D(y) \ \text{ for all } x, y\in A,
\end{equation}
then $D$ is
called a {\bf derivation} on $A$. Let $\Der(A)$ denote the set of all derivations of $A$.
More generally, a linear map $d:A\to A$ is called a {\bf derivation of weight $\lambda$}~\mcite{GK3} if
\begin{equation}
d(xy)=d(x)y+xd(y)+\lambda d(x)d(y), \quad \text{ for all } x, y\in A.
\end{equation}

We generalize the concepts of a Rota-Baxter operator and differential operator to $n$-algebras.
\begin{defn}
\mlabel{de:nlie} {\rm Let $\lambda\in k$ be fixed.
\begin{enumerate}
\item
An {\bf $n$-(nonassociative) algebra} over a field $\bfk$ is a
pair $(A, \langle , \cdots, \rangle)$ consisting of a vector space
$A$ over $\bfk$ and a multilinear
 multiplication
 $$\langle , \cdots, \rangle: A^{\otimes n} \rightarrow A.$$
\item
A {\bf derivation of weight $\lambda$} on an $n$-algebra $(A, \langle , \cdots, \rangle)$ is a
linear map $d: A\rightarrow A$ such that,
\begin{equation}
d(\langle x_1, \cdots,x_n\rangle) =\sum\limits_{\emptyset \neq I\subseteq [n]}\lambda^{|I|-1} \langle \check{d}(x_1),\cdots, \check{d}(x_i),\cdots, \check{d}(x_n)\rangle,
\mlabel{eq:dern}
\end{equation}
where
$\check{d}(x_i):=\check{d}_I(x_i):=\left \{\begin{array}{ll} d(x_i), & i\in I, \\ x_i, & i\not\in I\end{array} \right . \text{ for all } x_1,\cdots,x_n\in A.$
Then $A$ is called a {\bf differential $n$-algebra of weight $\lambda$.}
In particular, a {\bf differential $3$-algebra of weight $\lambda$} is a $3$-algebra $(A,\langle , , \rangle)$ with a linear map $d: A\to A$ such that
\begin{eqnarray}
d(\langle x_1,x_2,x_3\rangle )&=& \langle d(x_1),x_2,x_3\rangle +\langle x_1,d(x_2),x_3\rangle +\langle x_1,x_2,d(x_3)\rangle \notag \\
&&
+\lambda \langle d(x_1),d(x_2),x_3\rangle
+\lambda \langle d(x_1),x_2,d(x_3)\rangle
+\lambda \langle x_1,d(x_2),d(x_3)\rangle\\
&&+\lambda^2 \langle d(x_1),d(x_2),d(x_3)\rangle.
\notag
\end{eqnarray}
\item
A {\bf Rota-Baxter operator of weight $\lambda$} on $(A,\langle,\cdots,\rangle)$ is a linear map $P: A\rightarrow A$ such that
\begin{equation}
\langle P(x_1), \cdots, P(x_n)\rangle
=P\left( \sum\limits_{\emptyset \neq I\subseteq [n]}\lambda^{|I|-1} \langle \hat{P}(x_1), \cdots, \hat{P}(x_i), \cdots, \hat{P}(x_{n})\rangle\right),
\mlabel{eq:rbn}
\end{equation}
where
$\hat{P}(x_i):=\hat{P}_I(x_i):=\left\{\begin{array}{ll} x_i, & i\in I, \\ P(x_i), & i\not\in I \end{array}\right. \text{ for all } x_1,\cdots, x_n\in A.
$
Then $A$ is called a {\bf Rota-Baxter $n$-algebra of weight $\lambda$}.
In particular, a {\bf Rota-Baxter $3$-algebra} is a $3$-algebra $(A,\langle , , \rangle)$ with a linear map $P: A\to A$ such that
\begin{eqnarray}
\langle P(x_1),P(x_2),P(x_3)\rangle
&=& P\Big(
\langle P(x_1),P(x_2),x_3\rangle +\langle P(x_1),x_2,P(x_3)\rangle +\langle x_1,P(x_2),P(x_3)\rangle \notag \\
&&
+\lambda \langle P(x_1),x_2,x_3\rangle
+\lambda \langle x_1,P(x_2),x_3\rangle
+\lambda \langle x_1,x_2,P(x_3)\rangle
\mlabel{eq:rb3de}\\
&&
+\lambda^2 \langle x_1,x_2,x_3)\rangle\Big).
\notag
\end{eqnarray}
\end{enumerate}
}
\end{defn}

\begin{theorem}
Let $(A, \langle\ , \cdots, \rangle)$ be an $n$-algebra over $k$.
An invertible linear mapping $P: A\rightarrow A$ is a Rota-Baxter
operator of weight $\lambda$ on $A$ if and only if $P^{-1}$ is a
differential operator of weight $\lambda$  on $A$. \mlabel{thm:rbd}
\end{theorem}

\begin{proof}
If an invertible linear mapping $P: A\rightarrow A$ is a Rota-Baxter
operator of weight $\lambda$, then for $x_1, \cdots, x_n\in A$, set $y_i=P^{-1}(x_i)$. Then by Eq. (\mref{eq:rbn}) we have

\begin{eqnarray*}
P^{-1}(\langle x_1, \cdots, x_n\rangle)
&=& P^{-1}(\langle P(y_1), \cdots, P(y_n)\rangle) \\
&=& \sum\limits_{\emptyset \neq I\subseteq [n]}\lambda^{|I|-1}
\langle \hat{P}(y_1), \cdots, \hat{P}(y_i), \cdots,
\hat{P}(y_{n})\rangle\\
&=&\sum\limits_{\emptyset \neq I\subseteq
[n]}\lambda^{|I|-1} \langle \hat{P}P^{-1}(x_1), \cdots,
\hat{P}P^{-1}(x_i), \cdots, \hat{P}P^{-1}(x_{n})\rangle\\
&=& \sum\limits_{\emptyset \neq I\subseteq
[n]}\lambda^{|I|-1} \langle \breve{P^{-1}}(x_1), \cdots,
\breve{P^{-1}}(x_i), \cdots, \breve{P^{-1}}(x_{n})\rangle.
\end{eqnarray*}
Therefore, $P^{-1}$ is a derivation of weight $\lambda$.

Conversely, let $P$ be a derivation of weight $\lambda$. For
$x_1, \cdots, x_n\in A$, by Eq. (\mref{eq:dern}) we have

\begin{eqnarray*}
P(\langle P^{-1}(x_1), \cdots, P^{-1}(x_n)\rangle)&=&\sum\limits_{\emptyset \neq I\subseteq [n]}\lambda^{|I|-1} \langle \check{P}P^{-1}(x_1),\cdots, \check{P}P^{-1}(x_i),\cdots, \check{P}P^{-1}(x_n)\rangle\\
&=&\sum\limits_{\emptyset \neq I\subseteq
[n]}\lambda^{|I|-1} \langle\widehat{ P^{-1}}(x_1),\cdots,
\widehat{P^{-1}}(x_i),\cdots, \widehat{P^{-1}}(x_n)\rangle.
\end{eqnarray*}
Therefore,
$$\langle P^{-1}(x_1), \cdots,
P^{-1}(x_n)\rangle=P^{-1}\left(\sum\limits_{\emptyset \neq I\subseteq [n]}\lambda^{|I|-1}
\langle\widehat{ P^{-1}}(x_1),\cdots, \widehat{P^{-1}}(x_i),\cdots,
\widehat{P^{-1}}(x_n)\rangle\right).$$
This proves the result.
\end{proof}

\begin{prop}
Let $(A, \circ, P)$ be an associative Rota-Baxter (resp. differential) algebra of weight $\lambda$.  Define an $n$-ary
multiplication $\langle\, , \cdots, \,\rangle$ on $A$ by
$$\langle\, x_1, \cdots, x_n\,\rangle= x_1 \circ \cdots \circ x_n \ \text{ for all } x_1, \cdots, x_n\in
A.$$
Then $(A, \langle, \cdots, \rangle, P)$ is a Rota-Baxter
(resp. differential) $n$-algebra of weight $\lambda$.
\mlabel{pp:asn}
\end{prop}

\begin{proof}
We prove the case of Rota-Baxter algebras by induction on $n\geq 2$. The case of differential algebras is the same. There is nothing to prove when $n=2$. Suppose it is true for the case $n-1\geq 2$. Then

\begin{eqnarray*}
\lefteqn{\langle P(x_1), \cdots, P(x_n)\rangle=(P(x_1)\circ \cdots \circ
P(x_{n-1}))\circ P(x_n)}\\
&=&P(\sum\limits_{I\subseteq [n-1]}\lambda^{|I|-1}
\langle\hat{P}(x_1), \cdots, \hat{P}(x_i),\cdots,
\hat{P}(x_{n-1})\rangle) \circ P(x_n)\\
&=&P\Big\{P(\sum\limits_{I\subseteq [n-1]}\lambda^{|I|-1}
\langle\hat{P}(x_1), \cdots, \hat{P}(x_i),\cdots,
\hat{P}(x_{n-1})\rangle) \circ x_n\\
&&+\big(\sum\limits_{I\subseteq
[n-1]}\lambda^{|I|-1} \langle\hat{P}(x_1), \cdots,
\hat{P}(x_i),\cdots, \hat{P}(x_{n-1})\rangle\big) \circ P(x_n)\\
&&+\lambda \big(\sum\limits_{I\subseteq
[n-1]}\lambda^{|I|-1} \langle\hat{P}(x_1), \cdots,
\hat{P}(x_i),\cdots, \hat{P}(x_{n-1})\rangle\big) \circ x_n\Big\}.
\end{eqnarray*}
The first sum inside $P$ gives $\langle P(x_1),\cdots,P(x_{n-1}),x_n\rangle$. Together with the third sum inside $P$, we obtain
$$ \sum\limits_{\emptyset \neq I\subseteq [n], n\in I} \lambda^{|I|-1} \langle \hat{P}(x_1),\cdots, \hat{P}(x_n)\rangle.$$
The second sum inside $P$ is

$$\sum\limits_{\emptyset \neq I\subseteq [n], n\notin I}\lambda^{|I|-1}
\langle\hat{P}(x_1), \cdots, \hat{P}(x_i),\cdots,
\hat{P}(x_{n})\rangle.
$$
Therefore, $P$ is  a Rota-Baxter operator of weight $\lambda$ on the $n$-algebra $(A, \langle, \cdots,
\rangle)$.
\end{proof}

\section{Rota-Baxter $3$-Lie algebras from  Rota-Baxter Lie algebras and pre-Lie algebras}
\mlabel{sec:rbl3rbl}

In this section, we study the realizations of Rota-Baxter $3$-Lie algebras from  Rota-Baxter  Lie algebras in Section~\mref{ss:rbl3rbl} and from Rota-Baxter pre-Lie algebras in Section~\mref{ss:prelie3lie}.

\subsection{Rota-Baxter $3$-Lie algebras from Rota-Baxter Lie algebras}
\mlabel{ss:rbl3rbl}

By Definition~\mref{de:nlie}, suppose $(L, [, , ], P)$ is a Rota-Baxter $3$-Lie
algebra of weight $\lambda$. Then the linear map $P: L \rightarrow
L$ satisfies
\begin{eqnarray}
[P(x_1),P(x_2),P(x_3)] &=& P\Big( [P(x_1),P(x_2),x_3] +[
P(x_1),x_2,P(x_3)] +[ x_1,P(x_2),P(x_3)] \notag \\
&&+ \lambda[ P(x_1),x_2,x_3] +\lambda [
x_1,P(x_2),x_3] + \lambda[ x_1,x_2,P(x_3)]
\label{eq:3lie}
\\
&&+ \lambda^2[x_1,x_2,x_3)]\Big) \quad \text{ for all } x_1, x_2, x_3\in L.\notag
\end{eqnarray}

We recall the following result from Lie algebras to $3$-Lie algebras.
\begin{lemma}\cite{BBW} \mlabel{lem:bbw}
  Let $(L, [ , ])$ be a Lie algebra,
and let $f\in L^*:=\mathrm{Hom}(L,\bfk)$  satisfy $f([x, y])=0$ for $x, y\in L$. Then
$L$ is a $3$-Lie algebra with the multiplication
\begin{equation}
[x, y, z]_f:=f(x)[y, z]+f(y)[z, x]+f(z)[x, y] \quad  ~~ \text{for all }  x, y, z\in L.
\mlabel{eq:f3}
\end{equation}
\end{lemma}
We now combine this result with Rota-Baxter operators.
\begin{theorem} Let $(L, [ , ], P)$ be a
Rota-Baxter Lie algebra of weight $\lambda$, $f\in L^{\ast}$ satisfying $f([x,y])=0$ for all $x, y\in L$. Then $P$ is a Rota-Baxter operator on the $3$-Lie algebra $(L, [ , , ]_f)$ define in Eq.~$($\mref{eq:f3}$)$ if and only if $P$
satisfies
\begin{equation}f(x)[P(y),P(z)]+f(y)[P(z), P(x)]+f(z)[P(x),P(y)]\in \Ker
(P+\lambda Id_L) \ \text{ for all } x, y, z\in L, \end{equation}
where $Id_L:
L\rightarrow L$ is the identity map. \mlabel{thm:rb3lief}
\end{theorem}

\begin{proof}
By Lemma~\mref{lem:bbw}, $L$ is a $3$-Lie algebra with the multiplication $[ , ,
]_f$ defined in Eq.~(\mref{eq:f3}). Now for arbitrary $x, y, z\in L$,
\begin{eqnarray*}
\lefteqn{[P(x),P(y),P(z)]_f =f(P(x))[P(y),P(z)]+f(P(y))[P(z),P(x)]+f(P(z))[P(x),P(y)] } \\
&=&P\Big(f(P(x)([P(y),z]+[y, P(z)])+f(P(y))([P(z), x]+[z,
P(x)])\\
&&+f(P(z))([P(x),y]+[x, P(y)]) +\lambda(f(P(x))[y, z]+f(P(y))[z, x]+f(P(z))[x, y])\Big).
\end{eqnarray*}
Applying Eqs. (\mref{eq:rb}), (\mref{eq:f3}) and regrouping, we obtain
\begin{eqnarray*}
&&P\Big([P(x),P(y),z]_f+[P(x),y,P(z)]_f+[x,P(y),P(z)]_f\\
&&+\lambda ([P(x), y, z]_f+[x,P(y), z]_f+[x, y, P(z)]_f)+\lambda^2[x,
y, z]_f\Big)\\
&=&P\Big(f(P(x))([P(y),z]+[y,P(z)])+f(P(y))([z,P(x)] +[P(z),x])\\
&&+f(P(z))([P(x),y]+[x,P(y)])+\lambda
(f(P(x))[y,z]+f(P(y))[z,x]+f(P(z))[x,
y])\\
&&+\lambda (f(x)([P(y), z]+[z, P(y)]+\lambda
[y, z])+f(y)([z,P(x)]+[P(z), x])+\lambda [z, x])\\
&&+f(z)([P(x),y]+[x,P(y)]+\lambda [x,
y])+f(x)[P(y),P(z)]+f(y)[P(z),P(x)]+f(z)[P(x),P(y)]\Big)\\
&=&[P(x),P(y),P(z)]_f+(P+\lambda
Id_{L})(f(x)[P(y),P(z)]+f(y)[P(z),P(x)]+f(z)[P(x),P(y)]).
\end{eqnarray*}
Then the theorem follows.
\end{proof}

 \begin{coro} Let $(L, [ , ], P)$ be a
Rota-Baxter Lie algebra of weight zero, $f\in L^{\ast}$ satisfy
$f([x,y])=0$ for all $x, y\in L$. Then $P$ is a Rota-Baxter
operator on the $3$-Lie algebra $(L, [ , , ]_f)$ if and only if $P$ satisfies

\begin{equation}
[f(x)P(y)-f(y)P(x),z]+[f(y)P(z)-f(z)P(y),x]+[f(z)P(x)-f(x)P(z), y]\in \Ker  P^2
\mlabel{eq:rb3lief}
\end{equation}
for $x, y, z\in L$, where, $[ , , ]_f$ is defined in
Eq.(\mref{eq:f3}).

In particular, if  $P^2=0$, then for every $f\in L^*$ satisfying
$f([x,y])=0$, $P$ is a Rota-Baxter operator on the $3$-Lie algebra
$(L, [ , , ]_f)$. \mlabel{co:rb3lief}
\end{coro}

\begin{proof} Let $(L, [ , ], P)$ be a
Rota-Baxter Lie algebra of weight zero, $f\in L^{\ast}$ satisfying $f([x,y])=0$ for $x, y\in L$. By Eq.~(\mref{eq:rb}) we have

\begin{eqnarray*}
&&f(x)[P(y),P(z)]+f(y)[P(z),P(x)]+f(z)[P(x),P(y)]\\
&=&P^2([f(x)P(y)-f(y)P(x),z]+[f(y)P(z)-f(z)P(y),x]+[f(z)P(x)-f(x)P(z), y]).
\end{eqnarray*}
Then the corollary follows from Theorem~\mref{thm:rb3lief}.
\end{proof}

\subsection{Rota-Baxter $3$-Lie algebras from Rota-Baxter  pre-Lie algebras}
\mlabel{ss:prelie3lie}

In this section, we study the realizations of Rota-Baxter $3$-Lie
algebras by Rota-Baxter associative algebras and Rota-Baxter pre-Lie
algebras. First we recall some properties (cf. \cite{HuiB}).

Let $L$ be a vector space over a field $F$
with a bilinear product $\ast$ satisfying

\begin{equation} (x*y)\ast z-x\ast (y\ast z) =
(y\ast x)\ast z-y\ast (x\ast z) \ \text{ for all } x, y, z\in L.
\mlabel{eq:prelie}
\end{equation}
Then $(L,\ast)$ is called a {\bf pre-Lie algebra}.
It is obvious that all
associative algebras are pre-Lie algebras. For a pre-Lie algebra $L$, the commutator

\begin{equation}[x, y]_{\ast}: = x\ast y - y\ast x,
\mlabel{eq:skewsym}
\end{equation}
defines a Lie algebra $ G(L)=(L, [ . ]_{\ast})$, called the
{\bf sub-adjacent Lie algebra} of the pre-Lie algebra $L$.

If a linear mapping $P: L\rightarrow L$ is a Rota-Baxter operator of
weight $\lambda$ on a pre-Lie algebra $(L, \ast)$, that is, $P$
satisfies
$$ P(x)\ast P(y)=P(P(x)\ast y+x\ast P(y)+\lambda x\ast y) \ \text{ for all } x, y\in L,$$
then $P$ is a Rota-Baxter operator of weight $\lambda$ on its
sub-adjacent Lie algebra $G(L)=(L, [ , ]_{\ast}).$

The following facts on relationship among various Rota-Baxter algebras can be easily verified from definitions.
\begin{lemma}
\begin{enumerate}
\item Let $(L, \ast, P)$ be a  Rota-Baxter
pre-Lie algebra of weight $\lambda$. Then $(L, [ , ]_{\ast}, P)$
is a Rota-Baxter Lie algebra of weight $\lambda$,  where the
multiplication $[ , ]_{\ast}$ is defined in Eq.~$($\mref{eq:skewsym}$)$
\mlabel{it:rbliea}
\item
  If $(L, \ast, P)$ is a Rota-Baxter pre-Lie algebra of weight
$\lambda$. Then $(L, \cdot, P)$ is a Rota-Baxter pre-Lie algebra of
weight $\lambda$, where

\begin{equation} x \cdot y: = P(x)\ast y-y\ast P(x)+ \lambda x\ast y\quad  \text{ for all } x, y\in
L.
\mlabel{eq:circ1}
\end{equation}
\mlabel{it:rblieb}
\item
Let $(A, \cdot, P )$ be a Rota-Baxter commutative associative algebra of weight $\lambda$, $D$ be a derivation on the algebra
$(A, \cdot)$ that satisfies $DP=PD$. Then $(A, \ast, P)$ is a
Rota-Baxter pre-Lie algebra of weight $\lambda$, where
\begin{equation}x\ast
y=D(x)\cdot y \quad \text{ for all } x, y\in A.
\mlabel{eq:astd}
\end{equation}  Therefore by Item~\mref{it:rblieb}, we
get a Rota-Baxter Lie algebra $(A, [ , ]_{\ast})$, where
\begin{equation}[x, y]_{\ast}=D(x)\cdot y-D(y)\cdot x \quad \text{ for all } x, y\in A.
\mlabel{eq:brad}
\end{equation}
\mlabel{it:rblied}
\end{enumerate} \mlabel{lem:rblie}
\end{lemma}

\begin{theorem} \mlabel{thm:frb3lie}
Let $(L, \ast, P)$ be a Rota-Baxter
pre-Lie algebra of weight $\lambda$ and let $f\in L^{\ast}$ satisfying
$f(x\ast y-y\ast x)=0$ for $x, y\in L$. Define
\begin{equation}[ x,
y, z]_f=f(x)(y\ast z-z\ast y)+f(y)(z\ast x-x\ast z)+f(z)( x\ast y-y
\ast x) ~~ \text{ for all } x, y, z\in L.
\label{eq:braf}
\end{equation}
Then $P$
 is  a Rota-Baxter operator on $3$-Lie algebra $(L, [ , ,
]_f)$ if and only if
\begin{eqnarray}
\lefteqn{f(x)(P(y)\ast P(z)-P(z)\ast P(y))+f(y)(P(z)\ast P(x)-P(x)\ast P(z))}
\label{eq:fprelie}\\
&&+f(z)(P(x)\ast P(y)-P(y)\ast P(x))\in \Ker  (P+\lambda Id_L) \ \text{ for all } x, y, z\in
L. \notag
\end{eqnarray}
\end{theorem}

\begin{proof} By Lemma~\mref{lem:bbw}, for $f\in L^*$ satisfying $f(x\ast
y-y\ast x)=0$, $(L, [ , , ]_f)$ is a $3$-Lie algebra with the
multiplication in Eq.~(\mref{eq:braf}). By
Lemma~\mref{lem:rblie}.\mref{it:rbliea} and
Theorem~\mref{thm:rb3lief}, $P$ is a Rota-Baxter operator on the
$3$-Lie algebra $(L, [ , , ]_{f})$ if and only if $P$ satisfies
Eq.~(\mref{eq:fprelie}).
\end{proof}

\begin{theorem} \mlabel{thm:3lierb}
Let $(L, \ast, P)$ be a Rota-Baxter pre-Lie algebra of weight zero, $f\in L^*$ satisfy
\begin{equation}f(P(x)\ast y-y\ast P(x))=f(P(y)\ast x-x\ast P(y)) \quad \text{ for all } x, y\in L.
\mlabel{eq:3lierb1}
\end{equation}
Then $(L, [ ,  , ])$ is a $3$-Lie algebra with the multiplication
\begin{eqnarray}[x, y, z]&:=&(f(x)P(y)-f(y)P(x))\ast z-z\ast(f(x)P(y)-f(y)P(x))\notag\\
&&+(f(z)P(x)-f(x)P(z))\ast y-y\ast (f(z)P(x)-f(x)P(z)) \mlabel{eq:3lierb2}\\
&& +(f(y)P(z)-f(z)P(y))\ast x-x\ast
(f(y)P(z)-f(z)P(y)). \notag
\end{eqnarray}
Further, $P$ is a Rota-Baxter operator of weight zero on the $3$-Lie algebra $(L, [ , ,])$ if and only if $P$ satisfies
\begin{eqnarray}
&&f(x)(P^2(y)\ast P^2(z)-P^2(z)\ast
P^2(y))+f(y)(P^2(z)\ast P^2(x)-P^2(x)\ast P^2(z))
\mlabel{eq:3lierb3} \\
&&+f(z)(P^2(x)\ast P^2(y)-P^2(y)\ast P^2(x))=0 \quad \text{ for all } x, y, z\in L. \notag
\end{eqnarray}
\end{theorem}

\begin{proof} By Lemma~\mref{lem:rblie}.\mref{it:rblieb}, $(L, \circ)$ is a pre-Lie algebra with the
multiplication
$$\circ: L\otimes L\rightarrow L,  x\circ y= P(x)\ast
y-y\ast P(x) \text{ for all }  x, y\in L,$$ and $P$ is a Rota-Baxter operator on the pre-Lie algebra $(L, \circ)$.

If $f\in L^*$ satisfies $f(x\circ y-y\circ x)=0$, that is, $f$
satisfies Eq. (\mref{eq:3lierb1}), then by Lemma~\mref{lem:bbw}, $(L, [ , , ])$ is a $3$-Lie
algebra, where
\begin{eqnarray*}
[x, y, z]&=&f(x)(y\circ z-z\circ y)+f(y)(z\circ x-x \circ z)+f(z)(x\circ y-y\circ
x)\\
&=&f(x)(P(y)\ast z-z\ast P(y)-P(z)\ast y+y\ast P(z))\\
&&+f(y)(P(z)\ast x-x\ast
 P(z)-P(x)\ast z+z\ast P(x))\\
&&+f(z)(P(x)\ast y-y\ast P(x)-P(y)\ast x+x\ast P(y))\\
&=&(f(x)P(y)-f(y)P(x))\ast z-z\ast(f(x)P(y)-f(y)P(x))\\
&&+(f(z)P(x)-f(x)P(z))\ast y-y\ast (f(z)P(x)-f(x)P(z))\\
&&+(f(y)P(z)-f(z)P(y))\ast x-x\ast
(f(y)P(z)-f(z)P(y)).
\end{eqnarray*}
Therefore, Eq.~(\mref{eq:3lierb2}) holds.

By  Theorem~\mref{thm:frb3lie},  $P$ is a Rota-Baxter operator on the
$3$-Lie algebra $(L, [ , , ])$ if and only if $P$ satisfies

\begin{eqnarray*}
0&=&P(f(x)(P(y)\circ P(z)-P(z)\circ P(y))+f(y)(P(z)\circ P(x)-P(x)\circ P(z))\\
&&+f(z)(P(x)\circ P(y)-P(y)\circ P(x)))\\
&=&f(x)P(P^2(y)\ast P(z)-P(z)\ast P^2(y)-P^2(z)\ast P(y)+P(y)\ast P^2(z))\\
&&+f(y)P(P^2(z)\ast P(x)-P(x)\ast P^2(z)-P^2(x)\ast P(z)+P(z)\ast P^2(x))\\
&&+f(z)P(P^2(x)\ast P(y)-P(y)\ast P^2(x)-P^2(y)\ast P(x)+P(x)\ast P^2(y))\\
&=&f(x)(P^2(y)\ast P^2(z)-P^2(z)\ast P^2(y)+ f(y)(P^2(z)\ast P^2(x)-P^2(x)\ast P^2(z))\\
&&+f(z)(P^2(x)\ast P^2(y)-P^2(y)\ast P^2(x).
\end{eqnarray*}
This proves the second statement.
 \end{proof}

\section{Rota-Baxter $3$-Lie algebras from Rota-Baxter commutative associative algebras}
\mlabel{sec:com3lie}

Let $(A, \cdot)$ be a commutative
associative algebra, $ D$ in $\Der A$,  $f$ in $A^*$ satisfying
$$f(D(x)y)=f(xD(y)).$$ Then by Lemma~\mref{lem:bbw} and Lemma~\mref{lem:rblie}.\mref{it:rblied}, $(A, [ , , ]_{f, D})$  is a
$3$-Lie algebra, where

\begin{eqnarray}
\lefteqn{ [x, y, z]_{f,D}:=\begin{vmatrix}
f(x) & f(y) & f(z)  \\
D(x) & D(y) & D(z)  \\
x    & y    &  z    \\
\end{vmatrix}}
\label{eq:fdm}
\\
&=&f(x)(D(y)\cdot z-D(z)\cdot y)+f(y)(D(z)\cdot x-D(x)\cdot
z)+f(z)(D(x)\cdot y-D(y)\cdot x) \notag \\
&=&D(f(x)y-f(y)x)\cdot z+D(f(z)x-f(x)z)\cdot y+D(f(y)z-f(z)y)\cdot x
\notag
\end{eqnarray}
for $x, y, z\in A$.

\begin{theorem} \mlabel{thm:fdmk}
 Let $(A, \cdot, P)$ be a commutative
associative Rota-Baxter algebra of weight $\lambda$, $ D\in \Der A$ satisfying
$PD=DP$ and $f\in A^*$ satisfying $f(D(x)y)=f(xD(y))$. Then $P$ is a Rota-Baxter operator of weight $\lambda$ on the $3$-Lie algebra $(A, [ , , ]_{f, D})$ if and only if $P$ satisfies
\begin{equation}\begin{vmatrix}
f(x) & f(y) & f(z)  \\
DP(x) & DP(y) & DP(z)  \\
P(x)    & P(y)    & P(z)    \\
\end{vmatrix}\in \Ker  (P+\lambda Id_L), \ \text{ for all } x, y, z\in A.
\mlabel{eq:fdmk}
\end{equation}
\end{theorem}

\begin{proof} The result follows directly from Theorem~\mref{thm:rb3lief} and Lemma~\mref{lem:rblie}.\mref{it:rblied}.
\end{proof}

We prove a lemma before giving our next results on Rota-Baxter $3$-Lie algebras.

\begin{lemma}
Let $A$ be a commutative algebra. For a $3\times 3$-matrix $M$, we use the notation
$$M:=\left[\begin{matrix} x_1 & y_1 & z_1 \\ x_2 & y_2 & z_2\\ x_3 & y_3 & z_3 \end{matrix}\right]=\left[\begin{matrix}\vec{x} & \vec{y} &\vec{z}\end{matrix}\right]$$
and the corresponding determinant, where $\vec{x}, \vec{y}$ and $\vec{z}$ denote the column vectors. Let $P:A\to A$ be a Rota-Baxter operator of weight $\lambda$ and let $P(\vec{x}), P(\vec{y})$ and $P(\vec{z})$ denote the images of the column vectors. Then we have
\begin{equation}
\begin{vmatrix} P(\vec{x}) & P(\vec y) & P(\vec z)\end{vmatrix}
= P\left(\sum_{\emptyset\neq I\subseteq [3]} \lambda^{|I|-1}
\begin{vmatrix} \hat{P}(\vec{x}) & \hat{P}(\vec y) & \hat{P}(\vec z)\end{vmatrix}\right).
\mlabel{eq:rbmat}
\end{equation}
\mlabel{lem:rbmat}
\end{lemma}
\begin{proof}
By the definition of determinants and Proposition~\mref{pp:asn}, we have
\begin{eqnarray*}
\begin{vmatrix} P(\vec x) &P(\vec y) &P(\vec z) \end{vmatrix} &=& \sum_{\sigma\in S_3} \mathrm{sgn} (\sigma) P(x_{\sigma(1)}) P(y_{\sigma(2)})P(z_{\sigma(3)}) \\
&=& \sum_{\sigma \in S_3} \mathrm{sgn}(\sigma) P\left(\sum_{\emptyset \neq I\subseteq [3]} \lambda^{|I|-1} \hat{P}(x_{\sigma(1)}) \hat{P}(y_{\sigma(2)})\hat{P}(z_{\sigma(3)})\right) \\
&=& P\left(\sum_{\emptyset \neq I\subseteq [3]} \lambda^{|I|-1} \begin{vmatrix} \hat{P}(\vec x) & \hat{P}(\vec y) &\hat{P}(\vec z) \end{vmatrix}\right),
\end{eqnarray*}
as needed.
\end{proof}

Let $A$ be a commutative associative  algebra, $D_1, D_2,
D_3$ be  derivations on $(A, \cdot)$ satisfying $D_1D_2=D_2D_1$ for
$i, j=1, 2$. Then by \cite{Pozhi} $A$ is  a $3$-Lie algebra with the multiplication
\begin{equation} [x_1, x_2, x_3]:=\begin{vmatrix}
x_1    & x_2    &  x_3    \\
D_1(x_1) & D_1(x_2) & D_1(x_3)  \\
D_2(x_1) & D_2(x_2) & D_2(x_3)  \\
\end{vmatrix}=\begin{vmatrix} \vec x & D_1(\vec x) &D_2(\vec x)\end{vmatrix} ~~ \text{ for all } \vec x:=[x_1, x_2, x_3]^T\in A^3.
\mlabel{eq:d12m}
\end{equation}

\begin{theorem} \mlabel{thm:pd12}
Let $(A, P)$ be a  Rota-Baxter
commutative associative  algebra of weight $\lambda$,  $D_1, D_2$ be derivations on $(A, \cdot)$ satisfying $D_1D_2=D_2D_1$, $PD_i=D_iP$
for $i=1, 2$. Then $P$ is a Rota-Baxter operator of weight
$\lambda$ on the $3$-Lie algebra $(A, [ , , ])$, where $[ , , ]$ is defined by Eq.~$($\mref{eq:d12m}$)$.
\end{theorem}

\begin{proof} Let $x_1, x_2, x_3\in A$. Since $PD_1=D_1P$ and  $PD_2=D_2P,$ by Lemma~\mref{lem:rbmat} and Eq.~(\mref{eq:d12m}) we have
\begin{eqnarray*}
[P(x_1),P(x_2),P(x_3)]&=& \begin{vmatrix} P(\vec x)&D_1(P(\vec x))& D_2(P(\vec x))\end{vmatrix} \\
&=& \begin{vmatrix} P(\vec x) & P(D_1(\vec x)) &P(D_2(\vec x))\end{vmatrix} \\
&=& P\left(\sum_{\emptyset \neq I\subseteq [3]} \lambda^{|I|-1} \begin{vmatrix} \hat{P}(\vec x) & \hat{P}(D_1(\vec x)) &\hat{P}(D_2(\vec x))\end{vmatrix}\right)\\
&=& P\left(\sum_{\emptyset \neq I\subseteq [3]} \lambda^{|I|-1}
\begin{vmatrix} \hat{P}(\vec x) & D_1 (\hat{P}(\vec x)) & D_2(\hat{P}(\vec x))\end{vmatrix} \right)\\
&=& P\left(\sum_{\emptyset \neq I\subseteq [3]} \lambda^{|I|-1} [\hat{P}(x_1), \hat{P}(x_2), \hat{P}(x_3)]\right).
\end{eqnarray*}
This is what we need.
\end{proof}

Let $A$ be a commutative associative  algebra, $D_1, D_2,
D_3$ be  derivations on $A$ satisfying $D_iD_j=D_jD_i$ for
$i\neq j$, $i, j=1, 2, 3$. Then by \cite{F} $A$ is a $3$-Lie algebra with the multiplication
\begin{equation} [x, y, z]_D:=\begin{vmatrix}
D_1(x_1) & D_1(x_2) & D_1(x_3)  \\
D_2(x_1) & D_2(x_2) & D_2(x_3)  \\
D_3(x_1)& D_3(x_2) & D_3(x_3)    \\
\end{vmatrix} =\begin{vmatrix} D_1(\vec x) & D_2(\vec x)&D_3(\vec x)\end{vmatrix}, ~~ \text{ for all }\vec x=[x_1, x_2, x_3]^T\in A^3.
\mlabel{eq:d123m}
\end{equation}

\begin{theorem} Let $(A, P)$ be a  Rota-Baxter
commutative associative  algebra of weight $\lambda$,  $D_1, D_2,
D_3$ be derivations of $(A, \cdot)$ satisfying $D_iD_j=D_jD_i$ and
$PD_i=D_iP$   for $i, j=1, 2, 3$, $i\neq j$.  Then $P$ is a Rota-Baxter operator of weight $\lambda$ on the $3$-Lie algebra $(A, [ ,
, ]_D)$, where the multiplication $[ , , ]_D$ is defined by Eq.~$($\mref{eq:d123m}$)$.
\mlabel{thm:d123m}
\end{theorem}

\begin{proof}
The proof follows the same argument as the proof for Theorem~\mref{thm:pd12}.
\end{proof}

\section{Inheritance properties of Rota-Baxter $3$-Lie algebras}
\mlabel{sec:rb3inh}
In this section, we study the inheritance properties of Rota-Baxter $3$-Lie algebras. Such a property of Rota-Baxter Lie algebras plays an important role in their theoretical study and applications such as in integrable systems~\mcite{Bai,BGN,STS}. This is presented in Section~\mref{ss:rb3lie}. We also establish a similar property for Rota-Baxter Lie triple systems in Section~\mref{ss:trip}.

\subsection{Rota-Baxter $3$-Lie algebras constructed by  Rota-Baxter  $3$-Lie algebras}
\mlabel{ss:rb3lie}

Let $(L, [ , , ], P)$ be a Rota-Baxter $3$-Lie algebra of weight
$\lambda$. Using the notation in Eq.~(\mref{eq:rbn}), we define a ternary operation on $L$ by

\begin{eqnarray}
\lefteqn{[ x_1,x_2,x_3 ]_{P}=\sum_{\emptyset\neq I\subseteq [3]} \lambda^{|I|-1}[\hat{P}_I(x_1),\hat{P}_I(x_2),\hat{P}_I(x_3)]}\notag\\
&=& [P(x), P(y), z]+[P(x), y, P(z)]+[x,P(y), P(z)]\mlabel{eq:pderiv}\\
&&+\lambda[P(x), y, z]+\lambda[x, P(y), z]+\lambda[x, y, P(z)]+\lambda^2[x, y, z] \text{ for all } x, y, z\in L.
\notag
\end{eqnarray}

Then we have the following result.

\begin{theorem} Let $(L, [ , , ], P)$ be a Rota-Baxter
$3$-Lie algebra of weight $\lambda$. Then with $[\,,\,,\,]$ in Eq.~$($\mref{eq:pderiv}$)$, $( L, [ , ,
]_{P}, P)$ is a Rota-Baxter $3$-Lie algebra of weight
$\lambda$.
\mlabel{thm:pderiv}
\end{theorem}

\begin{proof} First we prove that $( L, [ , , ]_{P})$
is a $3$-Lie algebra. It is clear that $[ , , ]_{P}$ is multi-linear
and skew-symmetric.

Let $x_1, x_2, x_3, x_4, x_5\in L$. Denote $y_1=[x_1,x_2,x_3]_P, y_2=x_4, y_3=x_5$. Then by Eqs.~(\mref{eq:nlie}), (\mref{eq:rbn}) and (\mref{eq:pderiv}), we have
\begin{eqnarray*}
[[x_1,x_2,x_3]_P,x_4,x_5]_P &=& [y_1,y_2,y_3]_P \\
&=& \sum_{\emptyset \neq I\subseteq [3]} \lambda^{|I|-1} [\hat{P}_I(y_1),\hat{P}_I(y_2),\hat{P}_I(y_3)] \\
&=& \sum_{\emptyset \neq I\subseteq [3], 1\not\in I} \lambda^{|I|-1}[P(y_1),\hat{P}_I(y_2),\hat{P}_I(y_3)] +
\sum_{\emptyset \neq I\subseteq [3], 1\in I} \lambda^{|I|-1}[y_1,\hat{P}_I(y_2),\hat{P}_I(y_3)]
\\
&=&
\sum_{\emptyset \neq I \subseteq [3], 1\not\in I}
\lambda^{|I|-1}[[P(x_1),P(x_2),P(x_3)],\hat{P}_I(y_2),\hat{P}_I(y_3)]\\
&& +
\sum_{\emptyset \neq I\subseteq [3], 1\in I}
\lambda^{|I|-1}\left [ \sum_{\emptyset \neq J\subseteq [3]} \lambda^{|J|-1} [\hat{P}_J(x_1),\hat{P}_J(x_2),\hat{P}_J(x_3)], \hat{P}_I(y_2),\hat{P}_I(y_3)\right]
 \\
&=& \sum_{\emptyset \neq K \subseteq [5], K\cap [3]=\emptyset}
\lambda^{|K|-1}[[\hat{P}_K(x_1),\hat{P}_K(x_2),\hat{P}_K(x_3)], \hat{P}_K(x_4),\hat{P}_K(x_5)]\\
&&+ \sum_{\emptyset \neq I \subseteq [5], I\cap [3]\neq \emptyset} \lambda^{|I|-1}[[\hat{P}_K(x_1),\hat{P}_K(x_2),\hat{P}_K(x_3)], \hat{P}_K(x_4),\hat{P}_K(x_5)]\\
&=&
\sum_{\emptyset \neq I \subseteq [5]}
\lambda^{|I|-1}[[\hat{P}_K(x_1),\hat{P}_K(x_2),\hat{P}_K(x_3)], \hat{P}_K(x_4),\hat{P}_K(x_5)].
\end{eqnarray*}

Since $(L, [\,,\,,\,])$ is a 3-Lie algebra, for any given $\emptyset \neq I\subseteq [5]$, we have

\begin{eqnarray*}
\lefteqn{[[\hat{P}_K(x_1),\hat{P}_K(x_2),\hat{P}_K(x_3)], \hat{P}_K(x_4),\hat{P}_K(x_5)] = [[\hat{P}_K(x_1),\hat{P}_K(x_4),\hat{P}_K(x_5)], \hat{P}_K(x_2),\hat{P}_K(x_3)]}\\
&&+
[[\hat{P}_K(x_2),\hat{P}_K(x_4),\hat{P}_K(x_5)], \hat{P}_K(x_3),\hat{P}_K(x_1)]+
[[\hat{P}_K(x_3),\hat{P}_K(x_4),\hat{P}_K(x_5)], \hat{P}_K(x_1),\hat{P}_K(x_2)].
\end{eqnarray*}
Thus from the above sum, we conclude that $(L,[\,,\,,\,]_P)$ is a
3-Lie algebra.

Further we have
\begin{eqnarray*}
[P(x_1),P(x_2),P(x_3)]_P&=& \sum_{\emptyset \neq I \subseteq [3]} \lambda^{|I|-1} [\hat{P}_I(P(x_1)),\hat{P}_I(P(x_2)),\hat{P}_I(P(x_3))]\\
&=& \sum_{\emptyset \neq I \subseteq [3]} \lambda^{|I|-1} [P(\hat{P}_I(x_1)),P(\hat{P}_I(x_2)),P(\hat{P}_I(x_3))]\\
&=& \sum_{\emptyset \neq I \subseteq [3]} \lambda^{|I|-1} P\left([\hat{P}_I(x_1),\hat{P}_I(x_2),\hat{P}_I(x_3)]_P\right).
\end{eqnarray*}
This proves that $P$ is a Rota-Baxter operator on $(L, [\,,\,,\,]_P)$.
\end{proof}

\begin{theorem}  Let $(L, [ , , ], P)$ be  a Rota-Baxter $3$-Lie algebra
of weight $\lambda$. Let $d$ be a differential operator of weight $\lambda$ on $L$ satisfying $dP=Pd$. Then $d$
is a derivation of weight $\lambda$ on the  $3$-Lie algebra $(L, [ , , ]_{P})$,
where $[ , , ]_{P}$ is defined in Eq.~$($\mref{eq:pderiv}$)$.
\mlabel{thm:pderiv2}
 \end{theorem}

\begin{proof}
Let $x_1, x_2, x_3\in L$. Using the notation in
Eqs.~(\mref{eq:dern}) and (\mref{eq:rbn}), we have

\begin{eqnarray*}
d([x_1,x_2,x_3]_P)&=& \sum_{\emptyset\neq I\subseteq [3]} \lambda^{|I|-1} d [\hat{P}_I(x_1),\hat{P}_I(x_2),\hat{P}_I(x_3)]\\
&=& \sum_{\emptyset\neq I\subseteq [3]} \lambda^{|I|-1} \left(\sum_{\emptyset\neq J\subseteq [3]} \lambda^{|J|-1} [\check{d}_J\hat{P}_I(x_1),\check{d}_J\hat{P}_I(x_2) ,\check{d}_J\hat{P}(x_3)]\right)\\
&=& \sum_{\emptyset\neq J\subseteq [3]} \lambda^{|J|-1} \left(\sum_{\emptyset\neq I\subseteq [3]} \lambda^{|I|-1} [\hat{P}_I\check{d}_J(x_1),\hat{P}_I\check{d}_J(x_2) ,\hat{P}\check{d}_J(x_3)]\right)\\
&=& \sum_{\emptyset \neq J \subseteq [3]} \lambda^{|J|-1} [\check{d}_J(x_1),\check{d}_J(x_2),\check{d}_J(x_3)]_P.
\end{eqnarray*}
Therefore, $d$ is a derivation of weight $\lambda$ on the $3$-Lie algebra $(L, [ , , ]_{P})$.\end{proof}

\begin{coro} Let $(L, [ , , ])$ be a $3$-Lie algebra,
$d$ be a invertible derivation of $L$ of weight $\lambda$.
Then $(L, [x, y, z]_{d^{-1}})$ with $[\,,\,,\,]$ defined in Eq.~$($\mref{eq:pderiv}$)$ is a $3$-Lie algebra. Further
\begin{equation}[ x, y, z ]_{d^{-1}}=d([d^{-1}(x), d^{-1}(y),
d^{-1}(z)]) \text{ for all } x, y, z\in L,
\mlabel{eq:pderivinv}
\end{equation}
and $d$ is a
derivation of weight $\lambda$ on the $3$-Lie algebra $(L, [ , , ]_{d^{-1}})$.
\mlabel{co:pderivinv}
\end{coro}

\begin{proof}
By Theorem~\mref{thm:rbd}, $d^{-1}$ is a Rota-Baxter operator of weight $\lambda$ on the $3$-Lie algebra $(L, [ , , ])$.
Then by Theorem~\mref{thm:pderiv}, $d^{-1}$ is a Rota-Baxter operator on the $3$-Lie algebra $L$ equipped with the multiplication $[ x, y, z
]_{d^{-1}}$ defined in Eq.~(\mref{eq:pderiv}). By Eq.~(\mref{eq:rbn}) we have
$$[ x, y, z ]_{d^{-1}}=d (d^{-1}([x,y,z]_{d^{-1}})) =d[d^{-1}(x), d^{-1}(y), d^{-1}(z)],$$
as needed. The last statement follows from
Theorem~\mref{thm:pderiv2}.
\end{proof}

\begin{coro}  Let $(L, [ , ], P)$ be a Rota-Baxter Lie
algebra of weight $\lambda$ and let $f\in L^*$. Suppose $f$ and $P$ satisfy
$$f([x, y])=0, \quad (P+\lambda Id_L)(f(x)[P(y), P(z)]+f(y)[P(z), P(x)])+f(z)[P(x), P(y)])=0.$$
Define
\begin{eqnarray}
\lefteqn{[ x, y, z ]_{f,P}:=f(P(x))([P(y), z]+[y,
P(z)]+\lambda [y,z])+ f(P(y))([P(z), x]+[z, P(x)]+\lambda [z,
x])} \notag
\\
&&+f(P(z))([P(x), y]+[y, P(x)]+\lambda[x, y]) +f(x)([P(y), P(z)]+\lambda[P(y), z]+\lambda[y, P(z)]\label{eq:pderivf}\\
&&
+\lambda^2[y, z])+f(y)([P(z), P(x)]+\lambda[P(z), x]+\lambda[z, P(x)]+\lambda^2[z, x]) \notag\\
&&+f(z)([P(x), P(y)]+\lambda[P(x), y]+\lambda[x, P(y)]+\lambda^2[x,y])\ \text{ for all } x, y, z\in L. \notag
\end{eqnarray}
Then $(L, [ , , ]_{f,P}, P)$ is a Rota-Baxter $3$-Lie algebra of
weight $\lambda$. \mlabel{co:pderiv}
\end{coro}

\begin{proof}
By Theorem~\mref{thm:rb3lief},
$$[x,y,z]_f:=f(x)[y,z]+f(y)[z,x]+f(z)[x,y]$$
defines a $3$-Lie algebra on $L$ for which $P$ is a Rota-Baxter operator of weight
$\lambda$. Then by Theorem~\mref{thm:pderiv}, the derived ternary multiplication $[\,,\,,\,]_{f,P}$
from $[\,,\,,\,]:=[\,,\,,\,]_f$ defined in Eq.~(\mref{eq:pderiv}) also equips $L$ with a $3$-Lie algebra structure
for which $P$ is a Rota-Baxter operator of weight $\lambda$. By direct checking, we see that $[\,,\,,\,]_{f,P}$ thus obtained agrees with the one defined in Eq.~(\mref{eq:pderivf}).
 \end{proof}

Taking the case when $\lambda=0$, we obtain
\begin{coro}
\mlabel{co:pbrak}
 Let $(L, [ , ], P)$ be a Rota-Baxter Lie
algebra of weight zero, and $f\in L^*$.   If $f$ and $P$ satisfy
$$f([x, y])=0, ~ P(f(x)[P(y), P(z)]+f(y)[P(z), P(x)])+f(z)[P(x), P(y)])=0,$$
then $(L,  [ , , ]_{P}, P)$ is a a Rota-Baxter $3$-Lie algebra of weight zero, where
\begin{eqnarray}
[ x, y, z ]_{P}&=&f(P(x))([P(y), z]+[y, P(z)])+ f(P(y))([P(z), x]+[z, P(x)])\notag\\
&&+f(P(z))([P(x), y]+[y, P(x)]) +f(x)[P(y), P(z)] \mlabel{eq:pbrak}\\
&&+f(y)[P(z), P(x)])+f(z)[P(x), P(y)] \ \text{ for all } x, y, z\in L. \notag
\end{eqnarray}
\end{coro}

\begin{remark} For a Rota-Baxter $3$-Lie algebra $(L, [ , , ],
P)$ of weight zero, $L$ may not be a $3$-Lie algebra with the
multiplication
\begin{equation}[x, y, z]_1=[P(x), y, z]+[x, P(y), z]+[x, y, P(z)]  ~~ \text{ for all } x,
y , z\in L.\end{equation}
 For Example, let $L$ be a $3$-Lie algebra in the multiplication
 $$[x_1, x_2, x_3]=x_4, [x_1, x_2, x_4]=x_3,[x_1, x_3, x_4]=x_2, [ x_2, x_3,
x_4]=x_1,$$ where $\{x_1, x_2, x_3, x_4\}$ is a basis of $L$. Since
$D=ad(x_1, x_2)+ad(x_3, x_4)$ is an invertible derivation of $L$ and
$D^{-1}=D$, $D$ is a Rota-Baxter of weight zero on $(L, [ , , ])$.
And
$$[x_1, x_2, x_3]_1=[D(x_1), x_2, x_3]+[x_1, D(x_2), x_3]+[x_1, x_2, D(x_3)]=x_3,$$
$$[x_1, x_2, x_4]_1=[D(x_1), x_2, x_4]+[x_1, D(x_2), x_4]+[x_1, x_2, D(x_4)]=x_4,$$
$$[x_1, x_3, x_4]_1=[D(x_1), x_3, x_4]+[x_1, D(x_3), x_4]+[x_1, x_3, D(x_4)]=x_1,$$
$$[x_2, x_3, x_4]_1=[D(x_2), x_3, x_4]+[x_2, D(x_3), x_4]+[x_2, x_3, D(x_4)]=x_2.$$
Since $[[x_1, x_2, x_3]_1, x_2, x_4]_1=-x_2$, and
$$[[x_1, x_2,
x_4]_1, x_2, x_3]_1+[x_1, [x_2, x_2, x_4]_1, x_3]_1+[x_1, x_2, [x_3,
x_2, x_4]_1]_1=x_2,$$ $L$ is not a $3$-Lie algebra in the
multiplication $[ , , ]_1$.
\end{remark}

\subsection{Rota-Baxter Lie triple systems}
\mlabel{ss:trip}

 A {\bf Lie triple system}~\cite{Ja} is a vector space $L$ equipped with a ternary linear bracket
$\{ , , \}:$ $ L\otimes L\otimes L \rightarrow L,$ satisfying
$$\{x, y, y\}=0,$$
$$\{x, y, z\}+\{y, z, x\}+\{z, x, y\}=0,$$
$$\{\{x, y, z\}, a, b\}=\{\{x, a, b\}, y, z\}+\{x, \{y, a, b\}, z\}+\{x, y, \{z, a, b\}\} \ \text{ for all } x, y, z\in L.$$
It is important in the study of symmetric spaces.

A Lie triple system is a $3$-Lie algebra and thus it makes sense to define a Rota-Baxter Lie triple system.

\begin{theorem} Let $(L, [ , , ], P)$ be a Rota-Baxter
Lie triple system of weight $\lambda$. Define a ternary multiplication on $L$ by $[ , , ]_{P}: L\otimes L\otimes L \rightarrow L$ in Eq.~$($\mref{eq:pderiv}$)$.
Then $(L, [\, ,\, ,\, ]_{P}, P)$ is a Rota-Baxter Lie triple system.
\mlabel{thm:rbtrip}
\end{theorem}

\begin{proof} It is clear that $[ x, y, y
]_{P}=0$ and
$$[ x, y, z ]_{P}+[ y, z, x ]_{P}+[ z, x, y ]_{P}=0 \ \text{ for all } x, y, z\in L. $$
Then the theorem follows from Theorem~\mref{thm:pderiv}.
\end{proof}

If $(L, [ , ])$ is a Lie algebra, then $(L, [ , , ])$  is a Lie
triple system \cite{LWG}, where $[ , , ]$ is defined by
\begin{equation}[x, y, z]:=[x, [y, z]]\  \text{ for all } x, y, z\in L.
\mlabel{eq:rb3rb}
\end{equation}

\begin{theorem} Let $(L, [ , ], P)$ be a
Rota-Baxter Lie algebra of weight $\lambda$. Then $(L, [, ,], P)$ is
a Rota-Baxter  Lie triple system of weight $\lambda$.
\mlabel{thm:rb3rb}
\end{theorem}

\begin{proof} By Eqs.~(\mref{eq:rb3rb}) and (\mref{eq:rb3de}), we have
\begin{eqnarray*}
\lefteqn{[P(x),[P(y),P(z)]]=[P(x),P[y,P(z)]] +[P(x),P[P(y),z]]+\lambda[P(x),P[y,z]]}\\
&=&P[P(x),[y,P(z)]]+P[x,P[y,P(z)]]+\lambda P[x,[y,P(z)]]\\
&& +P[P(x),[P(y),z]]+P[x,P[P(y),z]]+\lambda P[x,[P(y),z]]\\
&& +\lambda P[P(x),[y,z]]+\lambda P[x,P[y,z]]+\lambda^2 P[x,[y,z]].
\end{eqnarray*}
Since $P[x,[P(y),P(z)]]=P[x,P[y,P(z)]]+P[x,P[P(y),z]]+\lambda
P[x,P[y,z]], $ we have
\begin{eqnarray*}
[P(x),P(y),P(z)]&=& [P(x),[P(y),P(z)]]\\
&=&P[P(x),[P(y),z]]+P[P(x),[y,P(z)]] +P[x,[P(y),P(z)]]\\
&&+\lambda P[P(x),[y,z]]+\lambda P[x,[P(y),z]]+\lambda P[x,[y,P(z)]]+\lambda^2 P[x,[y,z]]\\
&=&P[P(x),P(y),z]+P[P(x),y,P(z)]+P[x,P(y),P(z)]\\
&&+\lambda P[P(x),y,z]+\lambda P[x,P(y),z]+\lambda P[x,y,P(z)]+\lambda^2 P[x,y,z]
\end{eqnarray*}
for all $x, y, z\in L$.
Hence $P$ is a Rota-Baxter operator on the $3$-Lie triple system $(L,[\,,\,,\,])$.\end{proof}

\noindent
{\bf Acknowledgements. }
Ruipu Bai was supported in part by the Natural
Science Foundation of Hebei Province (A2010000194). Li Guo acknowledges support from NSF grant DMS~1001855.

\bibliography{}

\begin{thebibliography}{999999}

\bibitem{Ag2} M. Aguiar, Pre-Poisson algebras, {\em Lett. Math. Phys.} {\bf 54} (2000) 263-277.

\bibitem{HuiB1} H. An, C. Bai, From Rota-Baxter algebras to pre-Lie algebras.
arXiv:0711-1389v1[math-ph].

\bibitem{AI} J. A. de Azcarraga, J. M. Izquierdo, n-ary algebras: a review with
applications, {\it J. Phys. A: Math. Theor.} {\bf 43} (2010) 293001,
arXiv: 1005.1028 [math-ph].

\bibitem{BL} J. Bagger, N. Lambert, Gauge symmetry
           and supersymmetry of multiple $M2$-branes, {\it Phys. Pev. D}
           {\bf 77} (2008) 065008.

\bibitem{Bai} C. Bai, A unified algebraic approach to classical Yang-Baxter equation, {\em J. Phys. A} {\bf 40} (2007), 11073-11082.

\bibitem{BBGN} C. Bai, O. Bellier, L. Guo and X. Ni, Spliting of operations, Manin products and Rota-Baxter operators,  {\em IMRN} (2012). DOI: 10.1093/imrn/rnr266.

\bibitem{BGN} C. Bai, L. Guo and X. Ni, Generalizations of the classical Yang-Baxter equation and O-operators, {\em J. Math. Phys.} {\bf 52} (2011) 063515.

\bibitem{BGN2} C. Bai, L. Guo and X. Ni, Nonabelian generalized Lax pairs, the classical Yang-Baxter equation and PostLie algebras, {\em Comm. Math. Phys.} {\bf 297} (2010) 553-596.

\bibitem{BBW} R. Bai, C. Bai and J. Wang, Pealizations of 3-Lie algebras, {\it J. Math. Phys.,} {\bf 51}, 063505 (2010)

\bibitem{BHB} R. Bai, W. Han and C. Bai,  The generating index of an n-Lie algebra, {\it J. Phys. A: Math. Theor.} {\bf 44} (2011) 185201 (14pp).

\bibitem{BSZ}  R. Bai , C. Shen  and Y. Zhang, 3-Lie algebras with an ideal N,
{\it Linear Alg. Appl.} {\bf 431} (2009) 673-700.

\bibitem{BSZ1}  R. Bai, G. Song and Y. Zhang, On classification of $n$-Lie algebras, {\it Front. Math. China}, 2011, {\bf 6}(4): 581-606.

\bibitem{Ba} G. Baxter, { An analytic problem whose solution
    follows from a simple algebraic identity,}
    {\em Pacific J. Math.} {\bf 10} (1960), 731--742.

\bibitem{Ca} P. Cartier, { On the structure of free Baxter algebras,}
    {\em Adv. Math.} {\bf 9} (1972), 253-265.

\bibitem{CK} A. Connes, D. Kreimer, Hopf algebras,
    Renormalisation and Noncommutative Geometry, Comm. Math. Phys.{\bf  199}
  (1988) 203-242

\bibitem{Pozhi} A.S. Dzhumadil' daev, Identities and derivations for Jacobi
algebras, arXiv: math.RA/0202040v1.

\bibitem{EGK}  K. Ebrahimi-Fard, L. Guo and D. Kreimer, Spitzer's identity and the algebraic Birkhoff decomposition in pQFT, {\em J. Phys. A: Math. Gen.} {\bf 37} (2004), 11037-11052.

\bibitem{EGM} K. Ebrahimi-Fard, L. Guo and D. Manchon, Birkhoff type decompositions and the Baker-Campbell-Hausdorff recursion, {\em Comm. Math. Phys.} {\bf 267} (2006) 821-845.

\bibitem{F}  V.T. Filippov, $n-$Lie algebras,  {\it Sib. Mat.
           Zh.,} {\bf 26} (1985) 126-140.

\bibitem{Guw} L. Guo, WHAT IS a Rota-Baxter algebra, {\em Notice Amer. Math. Soc.} {\bf 56} (2009) 1436-1437.

\bibitem{Gub} L. Guo, Introduction to Rota-Baxter Algebra, International Press and Higher Education Press, 2012.

\bibitem{GK1} L. Guo and W. Keigher, {Baxter algebras and shuffle
    products}, {\em Adv. Math.}, {\bf 150} (2000), 117--149.

\bibitem{GK3} L. Guo and W. Keigher, {On differential Rota-Baxter
algebras}, {\em J. Pure Appl. Algebra}, {\bf 212} (2008),
  522-540.

\bibitem {GSZ} L. Guo, W. Sit and R. Zhang, {Differemtail type operators and Gr\"obner-Shirshov bases},   {\em J. Symolic Comput.} (2012).

\bibitem{GZ} L. Guo and B. Zhang, {Renormalization of multiple zeta values}, {\em J. Algebra}, {\bf 319} (2008), 3770--3809.

\bibitem{G} A. Gustavsson, Algebraic structures on parallel
M2-branes, arXiv: 0709.1260.

\bibitem{HCK} P. Ho, M. Chebotar and W. Ke, On skew-symmetric maps on
Lie algebras, {\it Proc. Royal Soc. Edinburgh A} {\bf 113} (2003)
1273-1281.

\bibitem{HHM} P. Ho, R. Hou and Y. Matsuo, Lie $3$-algebra and multiple
$M_2$-branes, arXiv: 0804.2110.

\bibitem{HIM} P. Ho, Y. Imamura, Y. Matsuo, $M2$ to $D2$ revisited, {\it JHEP},
0807:003, 2008. DOI:  10.1088/1126-6708/2008/07/003

\bibitem{Ja} N. Jacobson, Lie and Jordan triple systems {\it Amer. J. Math.} {\bf 71} (1949) 149¨C170.

\bibitem{HuiB} X.X. Li, D.P. Hou and C.M. Bai, Rota-Baxter operators on pre-Lie
algebras, {\it J. Nonlinear Math. Phy.}, {\bf 14} (2007), no. 2, 269-289.

\bibitem{MP} D. Manchon and S. Paycha, Nested sums of symbols and renormalised multiple zeta values, {\em Int. Math. Res. Papers} 2010 issue 24, 4628-4697 (2010).

\bibitem{N} Y. Nambu, Generalized Hamiltonian dynamics, {\it Phys. Pev. D} {\bf 7} (1973)
                2405-2412.

\bibitem{P} G. Papadopoulos, M2-branes, $3$-Lie algebras and
Plucker relations, arXiv: 0804.2662.

\bibitem{K} S. Kasymov, On a theory of $n$-Lie algebras, {\it Algebra i Logika} {\bf 26} (1987)
277-297.

\bibitem{L} W. Ling, On the structure of $n-$Lie
                  algebras, Dissertation, { University-GHS-Siegen, Siegn,} 1993.

\bibitem{LWG} W.G. Lister, A structure theory of Lie triple systems,  {\it Trans. Amer. Math. Soc.} {\bf 72} (1952) 217¨C242

\bibitem{Ro1} G.-C. Rota, { Baxter algebras and combinatorial
    identities I, II,} {\em Bull. Amer. Math. Soc.} {\bf 75}
    (1969), 325--329, 330--334.

\bibitem{Ro2} G.-C. Rota, { Baxter operators, an introduction,}
    In: ``Gian-Carlo Rota on Combinatorics, Introductory papers
    and commentaries", Joseph P.~S. Kung, Editor,
    Birkh\"{a}user, Boston, 1995.

\bibitem{STS} M. A. Semenov-Tian-Shansky,
    { What is a classical $r$-matrix?},
   {\em Funct. Ana. Appl.}, {\bf 17} (1983) 259-272.

\bibitem{T} L. Takhtajan, On foundation of the generalized Nambu mechanics,
{\it Comm. Math. Phys.} {\bf 160} (1994) 295-315.

\end{thebibliography}

\end{document}